\documentclass[journal]{IEEEtran}

\usepackage[utf8]{inputenc}
\usepackage[T1]{fontenc}
\usepackage{amsmath,amssymb,amsthm}
\usepackage{mathtools}
\usepackage{bm}
\usepackage{graphicx}
\usepackage{booktabs}
\usepackage{enumitem}
\usepackage[colorlinks=true,linkcolor=blue,citecolor=blue,urlcolor=blue]{hyperref}
\usepackage{cleveref}

\theoremstyle{plain}      
\newtheorem{theorem}{Theorem}
\newtheorem{lemma}{Lemma}
\newtheorem{proposition}{Proposition}
\newtheorem{corollary}{Corollary}
\theoremstyle{definition}  
\newtheorem{definition}{Definition}

\theoremstyle{remark}      
\newtheorem{remark}{Remark}

\newcommand{\E}{\mathbb{E}}
\newcommand{\Var}{\mathbb{V}}
\newcommand{\Prob}{\mathbb{P}}
\newcommand{\R}{\mathbb{R}}

\newcommand{\SINR}{\mathrm{SINR}}

\newcommand{\ergcap}{\mathsf{C}}
\newcommand{\hsp}[3]{\hspace{#1pt}#3\hspace{#2pt}}
\newlength{\myfigwidth}
\DeclareMathOperator{\InvGamma}{Inv\text{-}Gamma}

\title{Multi-User Diversity Scaling in Heavy-Tailed Fading}
\author{Yonathan~Murin, Ali~\"{O}zer~Ercan, Nariman Farsad%
\thanks{The authors are with Reality Labs, Meta Platforms, Inc (e-mail: moriny@meta.com; aliercan@meta.com; nfarsad@meta.com).}
\thanks{This work has been submitted to the IEEE for possible publication. Copyright may be transferred without notice, after which this version may no longer be accessible.}\vspace{-4mm}}
\date{}

\begin{document}
\maketitle

\begin{abstract}
Classical multi-user diversity theory predicts that throughput over Rayleigh fading channels grows as $\log_2\!\log_2 K$.
In this work, we demonstrate a fundamental shift in this scaling law under heavy-tailed composite fading.
Specifically, under Fisher--Snedecor $\mathcal{F}$ composite fading, the channel power acquires a regularly varying upper tail,
shifting extreme-value statistics from the Gumbel to the Fr\'echet domain.
We prove that the maximum SINR among $K$ users scales polynomially as $K^{1/m_s}$, where $m_s$ is the shadowing severity parameter,
leading to an ergodic capacity scaling of $\frac{1}{m_s}\log_2 K$. Crucially, this scaling persists in interference-limited Poisson networks,
where aggregate co-channel interference alters the scaling constant but not the exponent.
This polynomial gain is most relevant in severe-to-moderate shadowing ($m_s \le 3$), as encountered in
body-area networks, vehicular/industrial IoT, and dense indoor environments, where Fr\'echet asymptotics
overtake industry-standard lognormal models at practical user counts.
Finally, we establish the conditions necessary to harvest this gain (showing that proportional-fair scheduling under quasi-static shadowing reverts to Gumbel scaling)
and validate all analytical findings through Monte Carlo simulations, including MIMO random beamforming.
\end{abstract}



\vspace{-3mm}
\section{Introduction}\label{sec:intro}

Multi-user diversity allows a system to increase throughput without additional bandwidth or power by scheduling the strongest of $K$ users.
Quantifying this growth is essential for designing dense networks. The theoretical foundations of multi-user diversity and opportunistic scheduling were established in~\cite{Knopp1995, Viswanath2002}.
Subsequent works expanded on this by providing comprehensive capacity analyses~\cite{Alouini1999} and evaluating feedback requirements~\cite{Gesbert2004}.

Classical theory predicts a slow, double-logarithmic capacity growth for the best-user scheduler~\cite{Viswanath2002}.
This restrictive scaling applies to all standard multi-path fading models~\cite{AlBadarneh2018} and persists under MIMO random beamforming~\cite{Sharif2005,Sharif2007},
as their \emph{exponential} tail decay confines the extreme-value statistics to the Gumbel domain~\cite{Leadbetter1983,DeHaan2006,Embrechts1997}.
Importantly, extending these models to \emph{composite} fading (multi-path plus large-scale shadowing)~\cite{Suzuki1977} does not alter this baseline.
Conventional composite models, such as generalized-$\mathcal{K}$~\cite{AlAhmadi2010} and 3GPP log-normal shadowing,
possess exponential or sub-exponential tails, leaving the system in the conservative Gumbel domain~\cite{DeHaan2006,Embrechts1997}.

In contrast, the Fisher--Snedecor $\mathcal{F}$ composite fading model~\cite{Yoo2017,Badarneh2018}, which combines Nakagami-$m$ multi-path
with \emph{inverse-Gamma} shadowing, introduces a \emph{power-law} upper tail.
Empirically validated for severe indoor and body-area networks~\cite{Cotton2014,Yoo2019}, this heavy tail fundamentally shifts the extreme-value
statistics into the \emph{Fr\'{e}chet} domain. Unlike the logarithmic growth of the Gumbel domain, the maximum of $K$ i.i.d. samples from an
$\mathcal{F}$ distribution grows \emph{polynomially} as $K^{1/m_s}$, where the tail index is governed entirely by the shadowing parameter
$m_s$~\cite{DeHaan2006,Resnick2007,Leadbetter1983,Embrechts1997}.
Characterizing this fundamentally different multi-user diversity scaling law provides the primary motivation for this work.

A complementary line of research in stochastic geometry models co-channel interferers as a Poisson point process (PPP),
yielding aggregate interference with a totally skewed, $\alpha$-stable power-law tail~\cite{Win2009,Baccelli2009,Haenggi2009,Haenggi2012,Andrews2011,ElSawy2017}.
While this framework has been extensively applied to complex network architectures and SINR meta-distributions~\cite{Dhillon2012,Schilcher2016},
prior studies on interference-limited multi-user diversity still assume traditional light-tailed fading for the desired link.
Consequently, the explicit interaction between a heavy-tailed $\mathcal{F}$-fading \emph{signal} and heavy-tailed PPP \emph{interference},
and its impact on diversity scaling, remains unexplored.
Bridging this gap requires more than a direct application of textbook extreme-value theory. 
First, verifying that Breiman's lemma applies when the denominator is itself heavy-tailed ($\alpha$-stable interference) demands a careful analysis 
of the near-zero behavior of the interference distribution. 
Second, our proportional-fair scheduling analysis reveals the surprising finding that the polynomial gain is entirely neutralized under quasi-static shadowing.

In this work we make the following contributions:
\paragraph{Asymptotic Scaling Laws} We prove that under Fisher--Snedecor $\mathcal{F}$ fading, the maximum SINR scales as $K^{1/m_s}$,
shifting ergodic capacity growth from double-logarithmic to the polynomial $\frac{1}{m_s}\log_2 K$ regardless of the noise or interference regime (\Cref{thm:general-scaling}).

\paragraph{Invariance to Poisson Interference} We establish that aggregate interference from a PPP affects only the scaling constant, not the exponent,
as the SINR upper tail remains dominated by the composite fading (\Cref{thm:sinr-tail}).

\paragraph{Unified Analytical Framework} We derive closed-form expressions for outage probability and second-order statistics,
while providing a systematic EVT-based classification of composite fading models (\Cref{thm:outage}, \Cref{prop:second-order}, \Cref{tab:scaling-comparison}).

\paragraph{Crossover Analysis} We quantify the "crossover point" where $\mathcal{F}$-fading capacity growth overtakes industry-standard log-normal models,
validated through slope-based comparisons and Monte Carlo simulations (\Cref{eq:K-slope-cross}, \Cref{tab:crossovers}).

\paragraph{Scheduling Design Criterion} We prove that proportional-fair scheduling under quasi-static shadowing
reverts to Gumbel scaling (\Cref{prop:pf-static}), establishing a precise condition for harvesting the polynomial gain:
shadowing must decorrelate within the scheduling window.

The remainder of this paper is organized as follows: \Cref{sec:system} introduces the system model and preliminaries;
\Cref{sec:results} proves the main scaling results; \Cref{sec:numerical,sec:discussion} provide numerical validation and practical implications;
and \Cref{sec:conclusion} concludes the paper.

\vspace{-3mm}
\section{System Model and Preliminaries}\label{sec:system}

\vspace{-1mm}
\subsection{Notation}
We denote probability and expectation by $\Prob[\cdot]$ and $\E[\cdot]$, respectively.
The survival (complementary CDF) function of a random variable~$X$ is $\bar{F}_X(x) = \Prob[X > x]$.
We write $f(x) \sim g(x)$ as $x \to \infty$ to denote asymptotic equivalence, namely
$\lim_{x\to\infty} f(x)/g(x) = 1$. We write $\xrightarrow{d}$ for convergence
in distribution. The Gamma and Beta functions are denoted by $\Gamma(\cdot)$ and $B(\cdot,\cdot)$.
Vectors are denoted by bold-face letters (e.g. $\mathbf{x}$), where $\|\cdot\|$ is the Euclidean norm.
The distribution $\mathcal{CN}(0, \sigma^2)$ denotes circularly symmetric complex Gaussian with zero mean and variance~$\sigma^2$.

\vspace{-3mm}
\subsection{Network Architecture and Received Signal}\label{ssec:system-model}

We consider a single-cell downlink where a base station (BS) at location $\mathbf{u}_0$ serves $K$ users at locations $\mathbf{u}_k$ indexed by $k = 1, \ldots, K$.
The system operates over bandwidth $W$ in a time-slotted fashion using a shared frequency band.

\paragraph{Signal and Channel Models}
The BS transmits symbol $x$ (with $\E[|x|^2] \hsp{-2}{-2}{=} 1$) with power $P$. 
The received signal at user $k$ is $y_k \hsp{-2}{-2}{=} \sqrt{P} r_k^{-\eta/2} \tilde{h}_k x + n_k$,
where $\tilde{h}_k$ is the complex channel coefficient from the BS to user~$k$, 
$r_k \hsp{-2}{-2}{=} \|\mathbf{u}_0 \hsp{-1}{-1}{-} \mathbf{u}_k\|$, $\eta \hsp{-2}{-2}{>} 2$ is the path-loss exponent, and $n_k \hsp{-2}{-2}{\sim} \mathcal{CN}(0, \sigma^2)$ is
additive white Gaussian noise (AWGN) with $\sigma^2 \hsp{-2}{-2}{=} N_0 W$, and $N_0$ being the noise power spectral density.
The received signal power is $S_k \hsp{-2}{-2}{=} P r_k^{-\eta} h_k$,
where the composite fading gain $h_k \hsp{-2}{-2}{=} |\tilde{h}_k|^2$ is decomposed into independent components
$h_k \hsp{-2}{-2}{=} g_k \cdot s_k$~\cite{Suzuki1977,Yoo2017}:
\begin{itemize}
    \item $g_k \hsp{-2}{-2}{\sim} \Gamma(m, 1/m)$ is the Nakagami-$m$ small-scale fading with shape $m \hsp{-2}{-2}{\ge} 1$ and $\E[g_k] \hsp{-2}{-2}{=} 1$.
    \item $s_k \hsp{-2}{-2}{\sim} \InvGamma(m_s, m_s - 1)$ is the large-scale shadowing with shape $m_s \hsp{-2}{-2}{>} 1$ and $\E[s_k] \hsp{-2}{-2}{=} 1$.
\end{itemize}
The product $h_k$ follows a scaled Fisher--Snedecor $\mathcal{F}(2m, 2m_s)$ distribution~\cite{Yoo2017,Badarneh2018}.
Importantly, it exhibits a \emph{power-law upper tail}: $\Prob[h_k > x] \hsp{-2}{-2}{\sim} C_h x^{-m_s}$ as $x \hsp{-2}{-2}{\to} \infty$ (see \Cref{lem:F-tail}),
unlike the exponential tail $\Prob[g_k > x] \hsp{-2}{-2}{\sim} e^{-mx}$ of Nakagami-$m$ fading.

\paragraph{Co-channel Interference}
Co-channel interferers are modeled as a homogeneous PPP $\Phi \hsp{-2}{-2}{\subset} \R^2$ with intensity $\lambda$.
Each interferer $j \hsp{-2}{-2}{\in} \Phi$ at $\mathbf{x}_j$ transmits with the same power $P$ as the serving BS.
The aggregate interference power is $I_k \hsp{-2}{-2}{=} \sum_{j \in \Phi} P \|\mathbf{x}_j - \mathbf{u}_k\|^{-\eta} \tilde{g}_j$, 
where $\tilde{g}_j \hsp{-2}{-2}{\sim} \mathrm{Exp}(1)$~\cite{Win2009}.
Following~\cite{Win2009,Baccelli2009}, $I_k$ follows a totally skewed $\alpha$-stable distribution with stability index
$\alpha_I \hsp{-2}{-2}{=} 2/\eta$ and dispersion $\sigma_I^{\alpha_I} \hsp{-2}{-2}{=} \lambda \pi P^{\alpha_I} \,\E[\tilde{g}_j^{\alpha_I}] / \alpha_I$.
The fading $h_k$ and aggregate interference $I_k$ are independent since they arise from disjoint transmitter sets.

\paragraph{SINR and Scheduling}
The SINR of user $k$ is $\gamma_k \hsp{-2}{-2}{=} \frac{P r_k^{-\eta} h_k}{I_k + \sigma^2}$.
We consider both the interference-limited ($I_k \hsp{-2}{-2}{\gg} \sigma^2$) and noise-limited ($I_k \hsp{-2}{-2}{=} 0$) regimes.
The BS employs an opportunistic scheduler selecting $k^\star \hsp{-2}{-2}{=} \arg\max_{1 \le k \le K} \gamma_k$.
We analyze the maximum SINR $\gamma_{\max}^{(K)} \hsp{-2}{-2}{=} \max_{k} \gamma_k$ and the ergodic capacity:
\begin{equation}\label{eq:ergodic-cap}
    \ergcap(K) = \E\!\bigl[\log_2\!\bigl(1 + \gamma_{\max}^{(K)}\bigr)\bigr],
\end{equation}
as $K \hsp{-2}{-2}{\to} \infty$.

\vspace{-4mm}
\subsection{Mathematical Preliminaries}
\label{ssec:preliminaries}

The multi-user diversity scaling is determined by the \emph{tail behavior} of the per-user SINR distribution.
We recall three classical results that formalize this connection.
A distribution whose survival function decays as a power law is captured by the notion of \emph{regular variation}:

\begin{definition}[Regular variation~{\cite{Embrechts1997,Resnick2007}}]
\label{def:rv}
A measurable function $f \hsp{-2}{-2}{:} (0,\infty) \hsp{-2}{-2}{\to} (0,\infty)$ is \emph{regularly
varying} with index $\rho \hsp{-2}{-2}{\in} \R$, written $f \hsp{-2}{-2}{\in} \mathrm{RV}_\rho$, if
$f(tx)/f(t) \hsp{-2}{-2}{\to} x^\rho$ as $t \hsp{-2}{-2}{\to} \infty$ for every $x > 0$.  A survival
function $\bar{F} \in \mathrm{RV}_{-\alpha}$ with $\alpha > 0$ is said to
have a \emph{power-law tail} with index~$\alpha$.
\end{definition}

The Fisher--Tippett--Gnedenko theorem then dictates how the maximum of
$K$ i.i.d.\ samples behaves, depending on whether the parent distribution
has an exponential or power-law tail:

\begin{theorem}[Fisher--Tippett--Gnedenko~{\cite{DeHaan2006}}]
\label{thm:frechet}
  Let $X_1, X_2, \ldots, X_K$ be i.i.d.\ with survival function
  $\bar{F}(x) \hsp{-1}{-1}{\sim} A\, x^{-\alpha}$ for $\alpha \hsp{-1}{-1}{>} 0$, $A \hsp{-1}{-1}{>} 0$.
  Then, as the number of samples $K \hsp{-1}{-1}{\to} \infty$,
  \begin{equation}\label{eq:frechet-limit}
    \frac{\max_{1 \le k \le K} X_k}{a_K}
    \;\xrightarrow{d}\; \Phi_\alpha,
    \quad
    a_K = (AK)^{1/\alpha},
  \end{equation}
  where $\Phi_\alpha(x) \hsp{-1}{-1}{=} \exp(-x^{-\alpha})$ for $x \hsp{-1}{-1}{>} 0$ is the Fr\'echet distribution.
  In particular, $\max_{1 \le k \le K} X_k \hsp{-1}{-1}{\sim} (AK)^{1/\alpha}$ in probability - a polynomial rate.
  If instead $\bar{F}$ decays exponentially (e.g., Rayleigh, Nakagami-$m$), the limit is the Gumbel distribution,
  and $a_K$ grows only as $\log K$.
\end{theorem}

The distinction between these two regimes, Gumbel (exponential tail, $\log K$ scaling of the maximum)
versus Fr\'echet (power-law tail, $K^{1/\alpha}$ scaling), is the mathematical mechanism that drives the results of this paper.

Finally, since the SINR can be written as a \emph{product} $\gamma_k \hsp{-1}{-1}{=} h_k \cdot P r_k^{-\eta} / I_k$
of the heavy-tailed fading power~$h_k$ and a lighter-tailed factor, we need a tool that describes how
regular variation propagates through such products. This is done via Breiman's lemma:

\begin{lemma}[Breiman's lemma~{\cite{Breiman1965}}]\label{lem:breiman}
  Let $X$ and $Y$ be independent non-negative random variables (RVs) with $\bar{F}_X \hsp{-1}{-1}{\in} \mathrm{RV}_{-\alpha}$ and
  $\E[Y^{\alpha + \epsilon}] \hsp{-1}{-1}{<} \infty$ for some $\epsilon \hsp{-1}{-1}{>} 0$.
  Then $\bar{F}_{XY}(x) \hsp{-1}{-1}{\sim} \E[Y^\alpha]\,\bar{F}_X(x)$ as $x \hsp{-1}{-1}{\to} \infty$.
  That is, $XY$ inherits the tail index of~$X$ and the lighter-tailed factor~$Y$ affects only the tail constant, not the exponent.
\end{lemma}

The analysis proceeds by (i)~establishing the regular variation of the $\mathcal{F}$-distributed channel power (\Cref{lem:F-tail}),
(ii)~proving the SINR inherits this tail under interference (\Cref{thm:sinr-tail}),
and (iii)~applying the Fr\'{e}chet theorem to derive the scaling of $\gamma_{\max}^{(K)}$ and $C(K)$ (\Cref{thm:general-scaling}).

\vspace{-1.5mm}
\section{Main Results}\label{sec:results}

\subsection{The Source of Heavy Tails: $\mathcal{F}$-Distributed Fading}

We start with showing that the composite fading power $h_k$ has a regularly varying tail.
This property is what ultimately drives the polynomial diversity scaling.
The following Lemma formalizes this result:

\begin{lemma}[Upper tail of the $\mathcal{F}$ distribution]\label{lem:F-tail}
  Let $h \hsp{-1}{-1}{=} g \cdot s$ where $g \hsp{-1}{-1}{\sim} \Gamma(m, 1/m)$ and
  $s \hsp{-1}{-1}{\sim} \InvGamma(m_s, m_s - 1)$ with $m_s \hsp{-1}{-1}{>} 1$.  Then, as $x \hsp{-1}{-1}{\to} \infty$,
  \begin{equation}\label{eq:lem-F-tail}
    \Prob[h \hsp{-1}{-1}{>} x] \hsp{-1}{-1}{\sim}
    \frac{\Gamma(m + m_s)}{\Gamma(m)\,\Gamma(m_s)\,m_s}
    \left(\frac{m_s - 1}{m}\right)^{m_s}
    x^{-m_s}.
  \end{equation}
  That is, $\bar{F}_h \in \mathrm{RV}_{-m_s}$.
\end{lemma}

\begin{proof}
As $g \hsp{-2.5}{-2.5}{\sim} \Gamma(m, 1/m)$ and $1/s \hsp{-2.5}{-2.5}{\sim} \Gamma(m_s, 1/(m_s-1))$ are
independent, $h \hsp{-2}{-2}{=} g \hsp{-1}{-1}{\cdot} s$ is a ratio of two independent Gamma RVs,
which follows a scaled $F$-distribution with $2m$ and $2m_s$ degrees of freedom~\cite{Yoo2017}.
Its PDF is given by $f_h(x) \hsp{-2}{-2}{=} \frac{\Gamma(m+m_s)}{\Gamma(m)\Gamma(m_s)}
  \left(\frac{m}{m_s - 1}\right)^{\!m}
  \hspace{-1pt} \frac{x^{m-1}}{\bigl(1 + \frac{m}{m_s - 1}x\bigr)^{m+m_s}},
  x \hsp{-2}{-2}{>} 0$.
For $x \hsp{-2.5}{-2.5}{\to} \infty$, $\frac{m}{m_s-1}x$ dominates the denominator, 
and therefore $\bigl(1 + \tfrac{m}{m_s-1}x\bigr)^{m+m_s} \hsp{-2}{-2}{\sim} \bigl(\tfrac{m}{m_s-1}\bigr)^{m+m_s} x^{m+m_s}$.
Substituting and simplifying, one obtains $f_h(x) \hsp{-2}{-2}{\sim} \frac{\Gamma(m+m_s)}{\Gamma(m)\,\Gamma(m_s)} \left(\frac{m_s-1}{m}\right)^{m_s} x^{-m_s - 1}$.
The survival function then follows by direct integration:
\begin{align}
  \Prob[h > x]
  \hsp{-2}{-2}{\sim}
  \frac{\Gamma(m+m_s)}{\Gamma(m)\,\Gamma(m_s)\,m_s}
  \left(\frac{m_s-1}{m}\right)^{m_s}
  x^{-m_s},
\end{align}
yielding~\eqref{eq:lem-F-tail}.
\end{proof}

Without shadowing, Nakagami-$m$ fading possesses an exponential tail ($\bar{F}_h(x) \hsp{-1}{-1}{\sim} e^{-mx}$) belonging to the Gumbel domain.
The inverse-Gamma shadowing transforms this into a power law $x^{-m_s}$, shifting the distribution to the Fr\'{e}chet domain.
Consequently, smaller $m_s$ indicates heavier shadowing and a more pronounced diversity scaling.

\begin{remark}[Boundary case $m_s \hsp{-1}{-1}{\le} 1$]\label{rem:ms-boundary}
We require $m_s \hsp{-1}{-1}{>} 1$ to ensure a finite shadowing mean ($\mathbb{E}[s_k] \hsp{-1}{-1}{=} 1$).
While $m_s \hsp{-1}{-1}{=} 1$ implies $\gamma_{\max}^{(K)} \hsp{-1}{-1}{\sim} K$, the mean SINR diverges;
for $m_s \hsp{-1}{-1}{<} 1$, the mean shadowing power is undefined. These regimes are thus excluded.
\end{remark}

\vspace{-3mm}
\subsection{Robustness to Interference}

A fundamental tension exists between the power-law gain of the $\mathcal{F}$-distributed signal and the $\alpha_I$-stable fluctuations
of the PPP interference $I_k$. While the signal's heavy tail theoretically enhances multi-user diversity,
aggregate interference is also characterized by extreme spikes from nearby interferers.
This creates a "competition of extremes": if interference spikes are sufficiently intense, they may "drown out" the signal's tail and neutralize diversity gains.
The central analytical question is whether the SINR inherits the robust power-law properties of the signal or is limited by nearest-neighbor interference.
The following theorem establishes that interference affects only the scaling constant, not the tail index.

\begin{theorem}[SINR tail]\label{thm:sinr-tail}
  Under the system model of \Cref{ssec:system-model} with $m_s \hsp{-2}{-2}{>} 1$, the SINR
  of each user has a regularly varying tail in both the noise-limited and interference-limited regimes:
  \begin{equation}\label{eq:sinr-tail}
    \Prob[\gamma_k > x] \sim A_k \, x^{-m_s},
    \quad x \to \infty,
  \end{equation}
  where $A_k$ is a user-dependent constant. Denoting the tail constant from \Cref{lem:F-tail} by
  \begin{equation}\label{eq:Ch}
    C_h =
    \frac{\Gamma(m + m_s)}{\Gamma(m)\,\Gamma(m_s)\,m_s}
    \left(\frac{m_s - 1}{m}\right)^{m_s},
  \end{equation}
  the scaling constant is $A_k \hsp{-2}{-2}{=} C_h\,(P r_k^{-\eta}/\sigma^2)^{m_s}$ in the noise-limited regime and
  $A_k \hsp{-2}{-2}{=} C_h\,\E[(P r_k^{-\eta}/I_k)^{m_s}]$ in the interference-limited regime.
  For the Rayleigh fading case $m \hsp{-2}{-2}{=} 1$, $\Gamma(1+m_s)/(\Gamma(1)\,\Gamma(m_s)\,m_s) \hsp{-2}{-2}{=} 1$,
  so $C_h$ simplifies to $(m_s \hsp{-2}{-2}{-} 1)^{m_s}$.
\end{theorem}

\begin{proof}
In the noise-limited regime, $\gamma_k \hsp{-2}{-2}{=} (P r_k^{-\eta}/\sigma^2)\,h_k$ is a scalar multiple of~$h_k$.
By \Cref{lem:F-tail},
$\Prob[\gamma_k \hsp{-1}{-1}{>} x] \hsp{-2}{-2}{=} \Prob[h_k \hsp{-1}{-1}{>} x \sigma^2 / (P r_k^{-\eta})]
\hsp{-2}{-2}{\sim} C_h \bigl(P r_k^{-\eta}/\sigma^2\bigr)^{m_s} x^{-m_s}$, giving $A_k \hsp{-2}{-2}{=} C_h\,(P r_k^{-\eta}/\sigma^2)^{m_s}$.

For the interference-limited case, we write $\gamma_k \hsp{-2}{-2}{=} h_k \cdot (c_k / I_k)$ with $c_k \hsp{-2}{-2}{=} P r_k^{-\eta}$.
We must verify that $Y_k \hsp{-2}{-2}{=} c_k/I_k$ satisfies the moment condition of Breiman's lemma.
Although $I_k$ has a heavy \emph{upper} tail, its density near zero decays super-exponentially: $f_I(x) \hsp{-2}{-2}{\sim} x^{-1 + \alpha_I/(1-\alpha_I)} \exp(-c\,x^{-\alpha_I/(1-\alpha_I)})$ as $x \to 0^+$. Consequently, all negative moments are finite: $\E[I_k^{-q}] \hsp{-2}{-2}{<} \infty$ for every $q \hsp{-2}{-2}{>} 0$ , and in particular $\E[Y_k^{m_s + \epsilon}] \hsp{-2}{-2}{<} \infty$.
Breiman's lemma (\Cref{lem:breiman}) then gives~\eqref{eq:sinr-tail}.
\end{proof}

Physically, while large interference drives the \emph{lower} SINR tail, extreme SINR peaks arise solely from rare inverse-Gamma shadowing events ($s_k$).
Therefore, interference impacts only the scaling constant~$A_k$, determining the absolute SINR \emph{level} without altering the growth \emph{rate} of the maximum SINR with~$K$.

\begin{remark}[$\mathcal{F}$-fading on interfering links]\label{rem:f-interference}
If interferers also experience $\mathcal{F}$-fading (with shadowing index~$m_s'$) rather than Rayleigh,
the aggregate interference acquires a heavier upper tail, but the SINR tail index remains~$m_s$.
This follows because the negative moments $\E[I_k^{-q}]$ remain finite (the near-zero behavior of $I_k$
is governed by the path-loss geometry, not the per-link fading), so Breiman's lemma still applies.
The diversity scaling $K^{1/m_s}$ is thus robust to the fading model of the interfering links.
\end{remark}

\vspace{-3mm}
\subsection{The Main Scaling Law}

Having the SNR tail result from \Cref{thm:sinr-tail}, the diversity scaling follows directly from
the Fr\'echet theorem.  The following theorem provides the capacity scaling with $K$,
constituting the central contribution of the paper.

\begin{theorem}[General capacity scaling]\label{thm:general-scaling}
  Under the system model of \Cref{ssec:system-model} with $m \hsp{-2}{-2}{\ge} 1$,
  $m_s \hsp{-2}{-2}{>} 1$, $\eta \hsp{-2}{-2}{>} 2$, in both the noise-limited and interference-limited
  regimes:
  \begin{equation}\label{eq:sinr-general}
    \gamma_{\max}^{(K)} \sim (AK)^{1/m_s},
    \quad K \to \infty,
  \end{equation}
  \begin{equation}\label{eq:capacity-general}
    \ergcap(K) \sim \frac{1}{m_s}\log_2 K.
  \end{equation}
  The constant $A$ depends on $m$, $m_s$, $\eta$, and the network geometry but not on~$K$.
\end{theorem}

\begin{proof}
\emph{SINR scaling.}
\Cref{thm:sinr-tail} implies that the per-user SINR has a regularly varying tail: $\bar{F}_{\gamma_k} \hsp{-2}{-2}{\in} \mathrm{RV}_{-m_s}$.
Applying \Cref{thm:frechet} with $\alpha \hsp{-2}{-2}{=} m_s$ gives $\gamma_{\max}^{(K)} / a_K \hsp{-2}{-2}{\xrightarrow{d}} Z$ as $K \hsp{-2}{-2}{\to} \infty$,
where $a_K \hsp{-2}{-2}{=} (AK)^{1/m_s}$ and $Z$ follows the standard Fr\'echet distribution with shape~$m_s$.
In particular, $\gamma_{\max}^{(K)} \hsp{-2}{-2}{\sim} (AK)^{1/m_s}$ in probability, establishing~\eqref{eq:sinr-general}.

\emph{Capacity scaling.}
We show that the convergence in distribution carries over to the expectation of $\log_2(1 \hsp{-1}{-1}{+} \gamma_{\max}^{(K)})$.
Let $Z_K \hsp{-2}{-2}{=} \gamma_{\max}^{(K)} / a_K$, such that $Z_K \hsp{-2}{-2}{\xrightarrow{d}} Z$ and
\begin{equation}\label{eq:ergcap_decomp}
  \ergcap(K) \hsp{-2}{-2}{=} \log_2 a_K \hsp{-1}{-1}{+} \E\!\bigl[\log_2(1/a_K \hsp{-1}{-1}{+} Z_K)\bigr].
\end{equation}
Since $Z_K \hsp{-2}{-2}{\xrightarrow{d}} Z$ with $Z \hsp{-2}{-2}{>} 0$ a.s.\ (the Fr\'echet distribution is supported on
$(0,\infty)$) and $1/a_K \hsp{-2}{-2}{\to} 0$, 
the continuous mapping theorem~\cite[Theorem~2.3]{Billingsley1999} gives
$\log_2(1/a_K + Z_K) \hsp{-2}{-2}{\xrightarrow{d}} \log_2 Z$.
To go from convergence in distribution to convergence of expectations, we note that
$|\log_2(1/a_K + Z_K)| \hsp{-2}{-2}{\le} Z_K^\epsilon$ for any $\epsilon \hsp{-2}{-2}{>} 0$
and all sufficiently large~$K$, since $\log$ grows slower than any power.
Since $\E[Z_K^\epsilon] \hsp{-2}{-2}{\to} \E[Z^\epsilon] \hsp{-2}{-2}{=} \Gamma(1 \hsp{-1}{-1}{-} \epsilon/m_s) \hsp{-2}{-2}{<} \infty$ for
$\epsilon \hsp{-2}{-2}{<} m_s$~\cite[Corollary~1.2.10]{DeHaan2006}, the sequence is uniformly integrable,
giving $\E[\log_2(1/a_K \hsp{-1}{-1}{+} Z_K)] \hsp{-2}{-2}{\to} \E[\log_2 Z]$.
Substituting $a_K \hsp{-2}{-2}{=} (AK)^{1/m_s}$ into \Cref{eq:ergcap_decomp} we obtain
$\ergcap(K) \hsp{-2}{-2}{=} \frac{1}{m_s}\log_2(AK) \hsp{-1}{-1}{+} \E[\log_2 Z] \hsp{-1}{-1}{+} o(1) \hsp{-2}{-2}{\sim} \frac{1}{m_s}\log_2 K$,
establishing~\eqref{eq:capacity-general}.
\end{proof}

A nice aspect of this result is its universality: the scaling exponent $1/m_s$ depends \emph{only} on
the shadowing severity parameter and is invariant to the multi-path fading parameter~$m$, the path-loss
exponent~$\eta$, and whether the system operates in a noise- or interference-limited regime.
All of these parameters affect the constant~$A$, and therefore the absolute SINR level, but not the rate at which multi-user diversity grows.
For the canonical case of moderately heavy shadowing ($m_s \hsp{-2}{-2}{=} 2$), this yields the $\sqrt{K}$ scaling:
\begin{corollary}[$\sqrt{K}$ scaling]\label{cor:sqrtK}
  For $m_s \hsp{-2}{-2}{=} 2$: $\gamma_{\max}^{(K)} \hsp{-2}{-2}{\sim} \sqrt{AK}$ and
  $\ergcap(K) \hsp{-2}{-2}{\sim} \tfrac{1}{2}\log_2 K$.
\end{corollary}

To put this in perspective, increasing the number of users from $K \hsp{-2}{-2}{=} 100$ to $K \hsp{-2}{-2}{=} 1000$ yields a capacity increase of
$\frac{1}{2}\log_2(1000/100) \hsp{-2}{-2}{=} \frac{1}{2}\log_2 10 \hsp{-2}{-2}{\approx} 1.66$~bits/s/Hz
under $\mathcal{F}$-fading with $m_s \hsp{-2}{-2}{=} 2$, compared to only $\approx 0.58$~bits/s/Hz under Rayleigh fading.
A nearly threefold difference that continues to widen as $K$ grows.

\vspace{-3mm}
\subsection{Outage and Higher-Order Statistics}

The Fr\'echet scaling has several important implications beyond the ergodic capacity.

\paragraph{Outage probability}
System designers often require guarantees on the probability that the SINR exceeds a threshold~$\tau$.
Under Fr\'echet scaling, this probability is more favorable than in the Gumbel regime.

\begin{proposition}[Outage probability]\label{thm:outage}
  The outage probability for a target threshold~$\tau$, in the high-SINR regime ($\tau \hsp{-2}{-2}{\gg} A^{1/m_s}$, so that
  $A\,\tau^{-m_s} \hsp{-2}{-2}{\ll} 1$), satisfies
  \begin{equation}\label{eq:outage}
    \Prob\!\bigl[\gamma_{\max}^{(K)} < \tau\bigr] \sim \exp\bigl(-A\,K\,\tau^{-m_s}\bigr).
  \end{equation}
  The required number of users for target outage $\epsilon$ is
  $K^\star \hsp{-2}{-2}{=} (\tau^{m_s}/A)\ln(1/\epsilon)$, growing \emph{polynomially} in~$\tau$.
\end{proposition}

\begin{proof}
As the $K$ users have i.i.d.\ SINRs, the outage probability factors as
$\Prob\bigl[\gamma_{\max}^{(K)} \hsp{-1}{-1}{<} \tau\bigr] \hsp{-2}{-2}{=} \bigl(1 \hsp{-1}{-1}{-} \Prob[\gamma_1 \hsp{-1}{-1}{>} \tau]\bigr)^K$.
From \Cref{thm:sinr-tail}, $\Prob[\gamma_1 \hsp{-1}{-1}{>} \tau] \hsp{-2}{-2}{\sim} A\,\tau^{-m_s}$ for large~$\tau$.
Substituting and applying the approximation $(1 - x)^K \hsp{-2}{-2}{\approx} e^{-Kx}$ for $x \hsp{-2}{-2}{\ll} 1$ gives~\eqref{eq:outage}.
Setting $\exp(-A\,K\,\tau^{-m_s}) \hsp{-2}{-2}{=} \epsilon$ and solving for~$K$ one obtains $K^\star \hsp{-2}{-2}{=} \frac{\tau^{m_s}}{A}\,\ln\frac{1}{\epsilon}$.
\end{proof}

Under Rayleigh fading ($\Prob[\gamma > \tau] \hsp{-2}{-2}{\sim} e^{-\tau}$), the required user pool grows exponentially as
$K^\star_{\mathrm{Ray}} \hsp{-2}{-2}{=} e^\tau \ln(1/\epsilon)$. For $m_s \hsp{-2}{-2}{=} 2$, the ratio of required users between the two regimes is:
\begin{equation}
    \frac{K^\star_{\mathrm{Ray}}}{K^\star_{\mathcal{F}}} = \frac{A \, e^\tau}{\tau^2}.
\end{equation}
At an SINR target of $\tau \hsp{-2}{-2}{=} 10$ ($10$\,dB), this ratio is $\approx 220\,A$.
Thus, Rayleigh fading requires roughly two orders of magnitude more users to meet the identical outage target.
Practically, opportunistic schedulers in $\mathcal{F}$-fading environments can guarantee stringent outage constraints with substantially smaller user populations,
minimizing the required cell loading.

\paragraph{Second-Order Statistics and Infinite Variance}
The following proposition establishes how the Fr\'{e}chet moment structure impacts scheduler reliability.

\begin{proposition}[Moments of $\gamma_{\max}^{(K)}$]
\label{prop:second-order}
  As $\gamma_{\max}^{(K)} / (AK)^{1/m_s} \hsp{-2}{-2}{\xrightarrow{d}} Z \hsp{-2}{-2}{\sim} \mathrm{Fr\acute{e}chet}(m_s)$:
  \begin{enumerate}[label=(\roman*)]
    \item $\E[\gamma_{\max}^{(K)}] \hsp{-2}{-2}{\sim} (AK)^{1/m_s}\,\Gamma(1 \hsp{-1}{-1}{-} 1/m_s) \hsp{-2}{-2}{<} \infty$ for $m_s \hsp{-2}{-2}{>} 1$.
    \item $\E[(\gamma_{\max}^{(K)})^2] \hsp{-2}{-2}{<} \infty$ only when $m_s > 2$.
    \item For $m_s \hsp{-2}{-2}{\le} 2$ (including $m_s \hsp{-2}{-2}{=} 2$): $\Var[\gamma_{\max}^{(K)}] \hsp{-2}{-2}{=} \infty$.
  \end{enumerate}
\end{proposition}

\begin{proof}
The standard Fr\'echet distribution with shape $\alpha \hsp{-2}{-2}{>} 0$ has moments $\E[Z^r] \hsp{-2}{-2}{=} \Gamma(1 - r/\alpha)$
for $0 \hsp{-2}{-2}{<} r \hsp{-2}{-2}{<} \alpha$, and $\E[Z^r] \hsp{-2}{-2}{=} \infty$ for $r \hsp{-2}{-2}{\ge} \alpha$, see~\cite[Chapter~3]{Embrechts1997}.
In addition, from \Cref{thm:general-scaling} it follows that $\gamma_{\max}^{(K)} / (AK)^{1/m_s} \hsp{-2}{-2}{\xrightarrow{d}} Z$ with $\alpha = m_s$.

(i)~For $r \hsp{-2}{-2}{=} 1$: $\E[Z] \hsp{-2}{-2}{=} \Gamma(1 - 1/m_s)$, which is finite for $m_s \hsp{-2}{-2}{>} 1$.
From the convergence-of-expectations argument of \Cref{thm:general-scaling} it follows that $\E[\gamma_{\max}^{(K)}] \hsp{-2}{-2}{\sim} (AK)^{1/m_s}\,\Gamma(1 - 1/m_s)$.

(ii)~For $r \hsp{-2}{-2}{=} 2$: $\E[Z^2] \hsp{-2}{-2}{=} \Gamma(1 - 2/m_s)$, which is finite only when $1 - 2/m_s \hsp{-2}{-2}{>} 0$.

(iii)~As $\Var[\gamma_{\max}^{(K)}]$ requires $\E[(\gamma_{\max}^{(K)})^2] \hsp{-2}{-2}{<} \infty$, the variance is infinite whenever
$m_s \hsp{-2}{-2}{\le} 2$.
\end{proof}

\Cref{prop:second-order} reveals a fundamental magnitude-reliability tradeoff separated by the physical boundary $m_s \hsp{-2}{-2}{=} 2$.
In severe shadowing environments ($m_s \hsp{-2}{-2}{\le} 2$), such as body-area networks~\cite{Cotton2014}, dense urban canyons, and industrial clutter,
the expected SINR scales rapidly. However, infinite variance causes massive slot-to-slot fluctuations, making single-slot gains highly unreliable
since the coefficient of variation does not shrink with~$K$, complicating
jitter-sensitive applications such as V2X (vehicle-to-everything) or industrial IoT. Conversely, milder propagation conditions ($m_s \hsp{-2}{-2}{>} 2$),
like suburban macro cells or office environments~\cite{Yoo2017,Badarneh2018}, yield slower diversity scaling ($(1/m_s)\log_2 K$).
Yet, scheduler performance becomes predictable, characterized by finite variance and diminishing relative fluctuations as~$K$ grows.
In practice, the infinite-variance regime ($m_s \hsp{-2}{-2}{\le} 2$) can be tamed by rate truncation
(capping the scheduled rate at a maximum value) or by averaging over multiple consecutive slots;
both mechanisms bound the variance while preserving the $1/m_s$ scaling exponent, at the cost of a reduced constant.

\vspace{-2mm}
\subsection{Scheduling and Fairness}
\label{subsec:sched_fair}

Max-throughput scheduling achieves polynomial scaling, but proportional-fair (PF) scheduling~\cite{Viswanath2002} introduces a
fundamental tradeoff depending on the shadowing coherence time.

\begin{proposition}[PF scheduling under quasi-static shadowing]\label{prop:pf-static}
  Suppose the shadowing component $s_k$ is fixed (quasi-static) for each user across the scheduling window, while the multi-path component
  $g_k[t] \hsp{-2}{-2}{\sim} \Gamma(m, 1/m)$ varies independently across slots. The PF scheduler selects
  $k^\star[t] \hsp{-2}{-2}{=} \arg\max_{1 \le k \le K} \gamma_k[t] / \bar{R}_k[t]$, where $\bar{R}_k[t]$ is the exponentially weighted moving-average (EWMA)
  throughput of user~$k$. Then the effective multi-user diversity scaling reverts to the Gumbel regime:
  \begin{equation}\label{eq:pf-gumbel}
    \ergcap_{\mathrm{PF}}(K) \sim \log_2\!\log K.
  \end{equation}
\end{proposition}

\begin{proof}
With static $s_k$, user~$k$'s SINR at slot~$t$ is $\gamma_k[t] \hsp{-2}{-2}{=} \mu_k \cdot g_k[t]$,
where $\mu_k \hsp{-2}{-2}{=} c_k\,s_k$ with $c_k \hsp{-2}{-2}{=} P\,r_k^{-\eta}/(I_k + \sigma^2)$ is constant across slots.
Since $g_k[t]$ is i.i.d.\ across slots, the EWMA converges almost surely to a user-dependent constant
$\bar{R}_k^\infty \hsp{-2}{-2}{=} \E[\log_2(1 + \mu_k\,g_k)]$, which depends on~$\mu_k$ but not on~$t$.
The PF selection $k^\star[t] \hsp{-2}{-2}{=} \arg\max_k \log_2(1 + \mu_k\,g_k[t]) / \bar{R}_k^\infty$ is therefore determined
entirely by the fast-fading realizations $\{g_k[t]\}_{k=1}^K$.
For any user~$k$ and threshold~$\tau > 0$,
\begin{equation*}
  \Prob\!\left[\frac{\log_2(1 + \mu_k\,g_k)}{\bar{R}_k^\infty} > \tau\right]
  = \Prob\!\left[g_k > \frac{2^{\tau \bar{R}_k^\infty} - 1}{\mu_k}\right],
\end{equation*}
which decays exponentially in~$\tau$ since $g_k \hsp{-2}{-2}{\sim} \Gamma(m, 1/m)$ has an exponential tail.
Hence, the per-user PF metric has an exponential-type upper tail for every~$\mu_k \hsp{-2}{-2}{>} 0$, placing the maximum over~$K$ users in the Gumbel domain.
The achievable rate $R_{k^\star}[t] \hsp{-2}{-2}{=} \log_2(1 + \mu_{k^\star}\,g_{k^\star}[t])$ consequently scales as $\log_2\!\log K$.
\end{proof}

Conversely, if the shadowing coherence time matches the slot duration (time-varying $s_k$), the PF metric does not normalize out the inverse-Gamma tail,
recovering the Fr\'echet scaling $\frac{1}{m_s}\log_2 K$.
This delineates a critical tradeoff: harvesting polynomial scaling requires either tolerating short-term unfairness (max-throughput) or operating
in high-mobility/body-area networks~\cite{Cotton2014} where heavy-tailed shadowing de-correlates rapidly.
In quasi-static environments, PF scheduling successfully enforces fairness but fundamentally sacrifices heavy-tailed diversity by design.

\vspace{-3mm}
\subsection{Extension to MIMO Random Beamforming}

The single-antenna scaling extends naturally to multi-antenna base stations employing random beamforming~\cite{Sharif2005}.

\begin{proposition}\label{prop:mimo-rb}
Consider an $M$-antenna BS transmitting over $M$ random orthonormal beams. Let ${\bm{w}_i, \dots, \bm{w}_M}$ denote the M random orthonormal beamforming vectors.
For user~$k$ on beam~$i$ under isotropic Rayleigh fading ($m \hsp{-2}{-2}{=} 1$, so that $|\bm{h}_k^H \bm{w}_i|^2 \hsp{-2}{-2}{\sim} \mathrm{Exp}(1)$), the effective channel gain $|\bm{h}_k^H \bm{w}_i|^2 s_k$ follows the identical $\mathcal{F}$ distribution as the SISO case.
Assuming negligible inter-beam interference (\Cref{rem:mimo-interbeam}), the per-beam SINR inherits the $-m_s$ tail index.
Consequently, selecting the best user per beam yields $\gamma_{\max}^{(i)} \hsp{-2}{-2}{\sim} K^{1/m_s}$ for $i \hsp{-2}{-2}{=} 1, \dots, M$,
and the overall sum-rate achieves an $M$-fold scaling improvement:
\vspace{-1mm}
\begin{equation}\label{eq:mimo-sum-rate}
    \ergcap_{\mathrm{sum}}(K, M) \sim \frac{M}{m_s}\,\log_2 K.
\end{equation}
\vspace{-6mm}
\end{proposition}

\begin{proof}
Since $|\bm{h}_k^H \bm{w}_i|^2$ and $s_k$ are independent, with the former $\mathrm{Exp}(1)$ and the latter $\InvGamma(m_s, m_s-1)$,
their product follows the $\mathcal{F}(2, 2m_s)$ distribution by the same argument as \Cref{lem:F-tail}.
The per-beam extreme-value analysis is identical to the SISO case, and the $M$ beams contribute independently to the sum rate.
\end{proof}

\begin{remark}[Gumbel vs. Fr\'echet MIMO Scaling]\label{rem:mimo-gumbel}
While Rayleigh fading yields $\ergcap_{\mathrm{sum}} \hsp{-2}{-2}{\sim} M \log_2\!\log K$, the Fr\'echet regime achieves $(M/m_s)\log_2 K$.
Both maintain the $M$-fold multiplexing gain, but Fr\'echet elevates the base scaling from doubly logarithmic to logarithmic.
Consequently, for $M\hsp{-2}{-2}{=}4, m_s\hsp{-2}{-2}{=}2$, and $K\hsp{-2}{-2}{=}100$,
Fr\'echet delivers $\approx 13.3$\,bits/s/Hz versus Gumbel's $\approx 8.8$\,bits/s/Hz, a roughly 50\% capacity increase.
\end{remark}

\begin{remark}[Inter-beam Interference]\label{rem:mimo-interbeam}
\Cref{prop:mimo-rb} assumes external interference or noise dominates inter-beam leakage.
Generally, the per-beam SINR is $\gamma_{k,i} \hsp{-2}{-2}{=} |\bm{h}_k^H \bm{w}_i|^2 s_k / (\sum_{j \neq i} |\bm{h}_k^H \bm{w}_j|^2 s_k + I_k + \sigma^2)$.
If inter-beam interference dominates, $s_k$ cancels, reducing $\gamma_{k,i}$ to the spatial alignment ratio $R_{k,i} \hsp{-2}{-2}{=} |\bm{h}_k^H \bm{w}_i|^2 / \sum_{j \neq i} |\bm{h}_k^H \bm{w}_j|^2$.
For isotropic Rayleigh fading, $\Prob[R_{k,i} > t] \hsp{-2}{-2}{\sim} t^{-(M-1)}$ \cite{Sharif2005}.
Hence, the system remains in the Fr\'echet domain with two distinct regimes: if $m_s \hsp{-2}{-2}{<} M-1$, the heavier shadowing tail dominates,
yielding the $\frac{M}{m_s}\log_2 K$ scaling; if $m_s \hsp{-2}{-2}{\ge} M-1$, spatial alignment dictates a $\frac{M}{M-1}\log_2 K$ scaling.
Therefore, \Cref{prop:mimo-rb} holds whenever $m_s \hsp{-2}{-2}{<} M-1$ or external interference prevents $s_k$ cancellation.
\end{remark}

\vspace{-3mm}
\section{Numerical Results}\label{sec:numerical}
\vspace{-1mm}

In this section, we validate the analytical results via Monte Carlo simulations.
We evaluate $K \hsp{-2}{-2}{\in} [10, 5000]$ users experiencing independent composite fading with
Rayleigh multipath ($m \hsp{-2}{-2}{=} 1$) and varying shadowing shapes $m_s \hsp{-2}{-2}{\in} \{1.5, 2, 3, 5\}$.
In the noise-limited regime, the baseline SNR is $10$\,dB, while comparative lognormal shadowing uses $\sigma_{\mathrm{dB}} \hsp{-2}{-2}{=} 8$\,dB.
In the interference-limited regime, aggregate interference $I$ is shared per time slot among all users in the cell ($\gamma_k \hsp{-2}{-2}{=} h_k / I$).
Across independent trials, $I$ is generated as a totally skewed $\alpha$-stable variable (index $\alpha_I \hsp{-2}{-2}{=} 2/\eta \hsp{-2}{-2}{=} 0.5$
for path-loss $\eta \hsp{-2}{-2}{=} 4$, scale $\sigma_I \hsp{-2}{-2}{=} 0.1$) using the Chambers--Mallows--Stuck algorithm~\cite{CMS1976,Win2009}.
This shared-interference model preserves the conditional i.i.d.\ structure and the Fr\'echet scaling exponent~$m_s$.%
\footnote{Allowing spatially varying $I_k$ per user does not change the scaling exponent,
since \Cref{thm:sinr-tail} applies per-user with user-dependent~$A_k$.}
To manage computational load, trials scale inversely with $K$ from $20{,}000$ to $1{,}000$, with $95\%$ bootstrap confidence intervals ($500$ resamples) computed for all estimates.

\vspace{-3mm}
\subsection{Scaling of SINR and Ergodic Capacity}
\vspace{-1mm}

\begin{figure}[t]
\centering
\includegraphics[width=\myfigwidth]{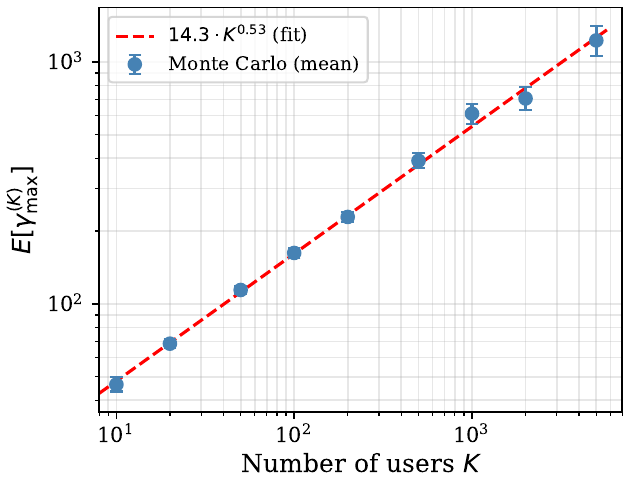}
\vspace{-2mm}
\caption{Monte Carlo validation of $\SINR_{\max}^{(K)}$ scaling for $m_s = 2$, $m = 1$, $\eta = 4$
(interference-limited, $1\text{k}$--$20\text{k}$ trials per point). Error bars show $95\%$ bootstrap confidence intervals.
The dashed line shows the fitted power law $14.3\,K^{0.53}$, confirming the theoretical $\sqrt{K}$ scaling
(fitted exponent $0.53$ vs.\ predicted $0.50$).}
\label{fig:sinr-scaling}
\vspace{-4mm}
\end{figure}

\Cref{fig:sinr-scaling} compares the Monte Carlo average of $\gamma_{\max}^{(K)}$ against a $C \hsp{-2}{-2}{\cdot} K^{0.53}$ power-law fit for $m_s = 2$.
The $0.53$ exponent closely matches the predicted $0.50$, with the slight bias arising from the Fr\'echet($2$) infinite-variance effect inflating the sample mean.
This monotonic growth across three decades of~$K$ confirms the $\sqrt{K}$ scaling derived in \Cref{cor:sqrtK}.

\Cref{fig:capacity-regimes} contrasts ergodic capacity scaling:
(a)~Rayleigh noise-limited ($C \hsp{-2}{-2}{\sim} \log_2\!\log K$);
(b)~$\mathcal{F}$-fading noise-limited ($C \hsp{-2}{-2}{\sim} \tfrac{1}{2}\log_2 K$);
and (c)~$\mathcal{F}$-fading interference-limited ($C \hsp{-2}{-2}{\sim} \tfrac{1}{2}\log_2 K$).
The two $\mathcal{F}$-fading curves share identical polynomial growth rates, differing only by a constant offset,
while the Rayleigh curve scales significantly slower.

\Cref{thm:general-scaling} explains the $K$-independent constant vertical gap,
$\Delta C \hsp{-2}{-2}{=} \frac{1}{m_s}\log_2 \frac{A_{\text{noise}}}{A_{\text{interf}}}$, between the $\mathcal{F}$-fading curves.
Both follow $C(K) \hsp{-2}{-2}{=} \frac{1}{m_s}\log_2 K \hsp{-1}{-1}{+} \frac{1}{m_s}\log_2 A \hsp{-1}{-1}{+} \E[\log_2 Z]$, diverging only in the parameter~$A$.
The SNR-dependent $A_{\text{noise}} \hsp{-2}{-2}{=} C_h\,(P r_k^{-\eta}/\sigma^2)^{m_s}$ exceeds the SIR-dependent
$A_{\text{interf}} \hsp{-2}{-2}{=} C_h\,\E[(P r_k^{-\eta}/I_k)^{m_s}]$ (with $C_h$ defined in~\eqref{eq:Ch}) because
dividing by small, deterministic $\sigma^2$ yields a larger value than dividing by random, heavy-tailed $I_k$.
Importantly, the $1/m_s$ slope remains invariant to the interference regime, as established in \Cref{thm:sinr-tail}.

\begin{figure}[t]
\centering
\includegraphics[width=\myfigwidth]{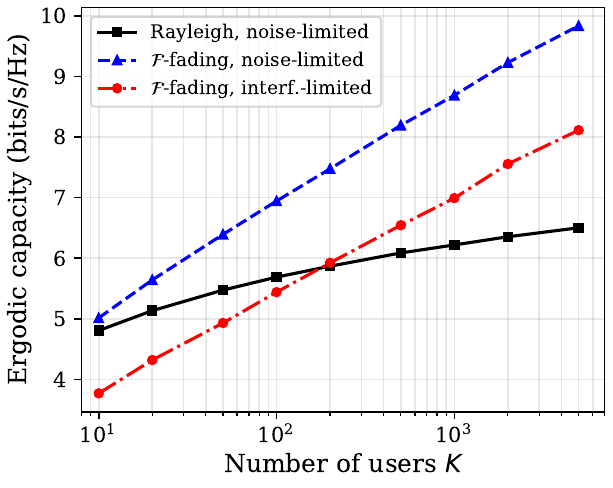}
\vspace{-2mm}
\caption{Ergodic capacity vs.\ $K$ under three channel/interference models.
The $\mathcal{F}$-fading curves ($m_s = 2$) grow as
$\tfrac{1}{2}\log_2 K$, dramatically outpacing the Rayleigh
$\log_2\!\log_2 K$ growth.  SNR${}= 10$\,dB for noise-limited curves.}
\label{fig:capacity-regimes}
\vspace{-3mm}
\end{figure}

\vspace{-4mm}
\subsection{Capacity Sensitivity to Shadowing Parameter $m_s$}

\Cref{fig:ms-sensitivity} plots simulated ergodic capacity versus $K$ across shadowing parameters $m_s$.
Confirming \Cref{thm:general-scaling}, heavier shadowing (smaller $m_s$) drives steeper $\frac{1}{m_s}\log_2 K$ capacity growth.
The fitted slopes - $0.63$ ($m_s=1.5$; theory $0.67$), $0.49$ ($m_s=2$; theory $0.50$), $0.35$ ($m_s=3$; theory $0.33$), and $0.27$ ($m_s=5$; theory $0.20$) -
match theory within $\sim 6\%$ for $m_s \hsp{-2}{-2}{\le} 3$, where the Fr\'echet regime establishes early ($K^\star \hsp{-2}{-2}{\le} 8$, per \Cref{tab:crossovers}).

Deviations occur at the extremes. For $m_s \hsp{-2}{-2}{=} 5$, lighter shadowing yields a thinner inverse-Gamma tail,
delaying extreme Fr\'echet events to $K^\star \hsp{-2}{-2}{\approx} 256$; this keeps the simulation in a pre-asymptotic (Gumbel-like) regime,
biasing the slope upward.
Conversely, for heavy shadowing ($m_s \hsp{-2}{-2}{\le} 2$), infinite Fr\'echet variance slows sample mean convergence, causing a slight downward bias.

Finally, the Rayleigh baseline ($\log_2\!\log K$ growth) is overtaken by the $m_s \hsp{-2}{-2}{\le} 3$ cases within the simulated range.
For $m_s \hsp{-2}{-2}{=} 5$, however, $\mathcal{F}$-fading starts lower due to reduced mean channel quality;
despite its asymptotically faster $\tfrac{1}{5}\log_2 K$ slope, absolute crossover with Rayleigh requires $K$ well beyond the simulated range.

\begin{figure}[t]
\centering
\includegraphics[width=\myfigwidth]{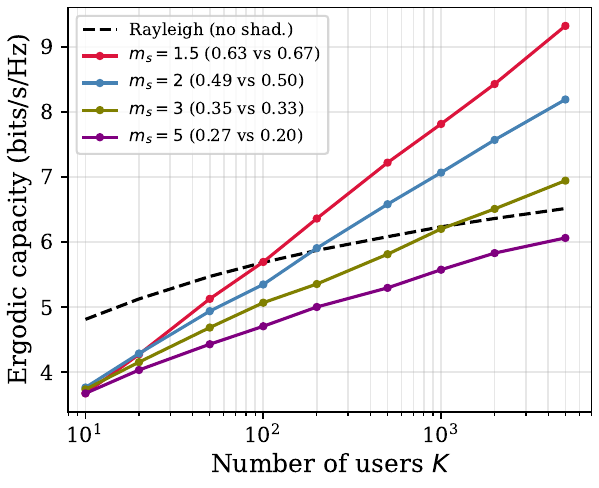}
\vspace{-2mm}
\caption{Ergodic capacity vs.\ $K$ for varying shadowing $m_s \in \{1.5, 2, 3, 5\}$ (interference-limited, $m=1$, $\eta=4$).
Heavier shadowing (smaller $m_s$) drives steeper $\frac{1}{m_s}\log_2 K$ growth compared to the $\log_2\!\log K$ Rayleigh baseline (black dashed).}
\label{fig:ms-sensitivity}
\vspace{-2mm}
\end{figure}

\vspace{-3mm}
\subsection{Comparison with Lognormal Shadowing}

A key practical question is whether the polynomial scaling persists under the industry-standard lognormal shadowing model.
\Cref{fig:lognormal-comparison} compares the ergodic capacity for four models in the noise-limited regime:
(a)~Rayleigh (no shadowing),
(b)~Nakagami-$m$ $\times$ Lognormal with $\sigma_{\mathrm{dB}} = 8$\,dB (a standard 3GPP value),
(c)~$\mathcal{F}$-fading with $m_s = 2$, and
(d)~$\mathcal{F}$-fading with $m_s = 1.5$.

\begin{figure}[t]
\centering
\includegraphics[width=\myfigwidth]{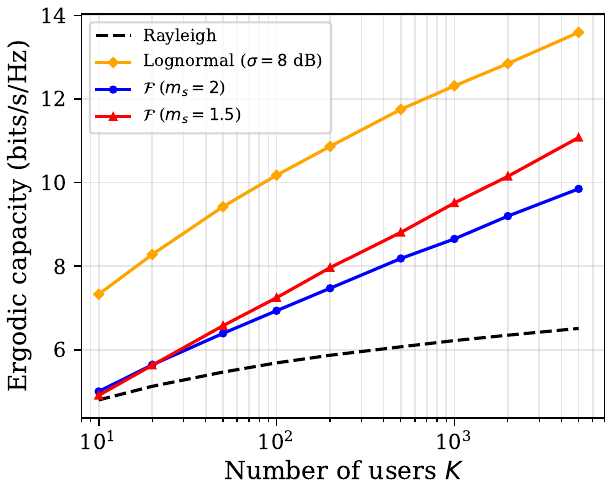}
\vspace{-2mm}
\caption{Capacity comparison: lognormal shadowing ($\sigma = 8$\,dB) vs.\ $\mathcal{F}$-fading (noise-limited, $m = 1$, SNR${}=10$\,dB).
Lognormal achieves higher absolute capacity for $K \le 5000$ due to large variance, but eventually slows to $\log_2\!\log K$ (Gumbel domain).
Conversely, $\mathcal{F}$-fading sustains $\frac{1}{m_s}\log_2 K$ asymptotic growth.}
\label{fig:lognormal-comparison}
\vspace{-2mm}
\end{figure}

At finite user counts (e.g., $K \hsp{-2}{-2}{\le} 5000$), lognormal shadowing ($\sigma_{\mathrm{dB}} \hsp{-2}{-2}{=} 8$\,dB) achieves higher absolute
capacity than $\mathcal{F}$-fading ($m_s \hsp{-2}{-2}{=} 2$) due to its large variance, which creates substantial diversity.
However, their asymptotic scaling rates differ fundamentally. The lognormal model belongs to the Gumbel domain~\cite[Sec.~3.3]{Embrechts1997};
the maximum of $K$ i.i.d.\ copies scales as $\exp(\sigma\sqrt{2\ln K})$, where $\sigma \hsp{-2}{-2}{=} \sigma_{\mathrm{dB}} \ln 10 / 10$.
Thus, its ergodic capacity grows sub-logarithmically:
\vspace{-1mm}
\begin{equation}\label{eq:C-LN}
  \ergcap_{\mathrm{LN}}(K) \approx \frac{\sigma\sqrt{2\ln K}}{\ln 2} + \text{const}.
\end{equation}
\vspace{-1mm}

While the lognormal capacity slope decays as $\sigma/\sqrt{2\ln K}$, the $\mathcal{F}$-fading slope remains a constant $1/m_s$ (\Cref{thm:general-scaling}).
Setting $\sigma/\sqrt{2\ln K} = 1/m_s$ and solving yields the \emph{slope crossover} user count:
\vspace{-1mm}
\begin{equation}\label{eq:K-slope-cross}
  K_{\mathrm{slope}} = \exp\Bigl(\frac{\sigma_{\mathrm{dB}}^2\, m_s^2\, \ln^2\!10}{200}\Bigr).
\end{equation}
\vspace{-1mm}

\begin{table}[t]
\centering
\caption{Slope ($K_{\mathrm{slope}}$) and Absolute ($K_{\mathrm{abs}}$) Capacity Crossovers}
\label{tab:crossovers}
\vspace{-2mm}
\footnotesize
\begin{tabular}{@{}ccrr@{}}
\toprule
$\sigma_{\mathrm{dB}}$ & $m_s$ & $K_{\mathrm{slope}}$ & $K_{\mathrm{abs}}$ \\
\midrule
$4$ & $1.5$ & $3$ & $\approx 200$ \\
$4$ & $2$ & $6$ & $\approx 1{,}500$ \\
$4$ & $3$ & $46$ & -- \\
$4$ & $5$ & $40{,}000$ & -- \\
\midrule
$8$ & $1.5$ & $46$ & $> 20{,}000$ \\
$8$ & $2$ & $890$ & $> 20{,}000$ \\
$8$ & $3$ & $4.3 \times 10^{6}$ & -- \\
\bottomrule
\end{tabular}
\vspace{-2mm}
\end{table}

\Cref{tab:crossovers} evaluates $K_{\mathrm{slope}}$ alongside $K_{\mathrm{abs}}$, the \emph{absolute capacity} crossover where
$\ergcap_{\mathcal{F}} \hsp{-2}{-2}{>} \ergcap_{\mathrm{LN}}$.
For moderate shadowing ($\sigma_{\mathrm{dB}} \hsp{-2}{-2}{=} 4$\,dB), $\mathcal{F}$-fading overtakes the lognormal rate and absolute capacity at practical user counts.
However, for the standard 3GPP value ($\sigma_{\mathrm{dB}} \hsp{-2}{-2}{=} 8$\,dB), strong lognormal shadowing delays the crossover;
the lognormal model maintains a higher absolute capacity beyond $K \hsp{-2}{-2}{=} 20{,}000$.
Although the $\mathcal{F}$-fading curve will eventually overtake it as the lognormal slope approaches zero,
this absolute crossover requires user counts well beyond practical simulated ranges.

\vspace{-2mm}
\subsection{Proportional-Fair Scheduling Validation}

\Cref{fig:pf-scheduling} validates \Cref{prop:pf-static} by comparing the ergodic capacity under three
scheduling scenarios with $\mathcal{F}$-fading ($m_s \hsp{-2}{-2}{=} 2$, $m \hsp{-2}{-2}{=} 1$, interference-limited):
(a)~max-throughput scheduling (select $\arg\max_k \gamma_k$), (b)~PF scheduling with static shadowing ($s_k$ fixed per user), and
(c)~PF scheduling with time-varying shadowing ($s_k$ redrawn each slot).
The PF window is set to $T_c \hsp{-2}{-2}{=} 5K$ to ensure each user's moving average converges.

\begin{figure}[t]
\centering
\includegraphics[width=\myfigwidth]{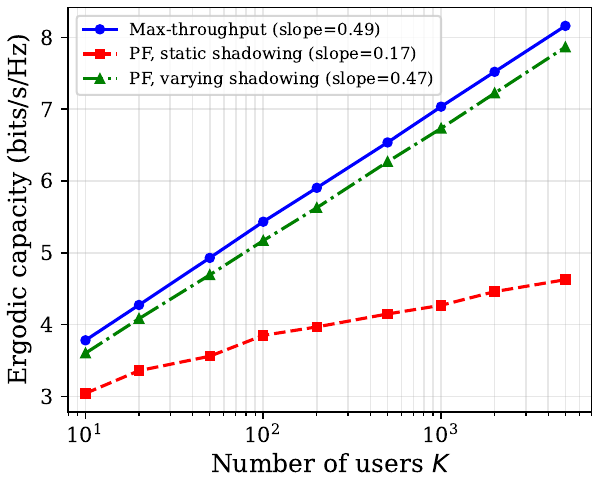}
\vspace{-2mm}
\caption{Ergodic capacity vs.\ $K$ under three scheduling scenarios
($m_s = 2$, $\eta = 4$, interference-limited, PF window $T_c = 5K$).
Max-throughput and PF with varying shadowing both exhibit Fr\'echet scaling
(slopes $0.49$ and $0.47$), while PF with static shadowing reverts to
near-flat Gumbel growth (slope $0.17$), confirming \Cref{prop:pf-static}.}
\label{fig:pf-scheduling}
\vspace{-3mm}
\end{figure}

The max-throughput curve (slope $0.49$) and PF with varying shadowing (slope $0.47$) both closely track the
theoretical $\frac{1}{m_s}\log_2 K$ scaling, confirming that the Fr\'echet diversity gain is accessible
when shadowing varies across scheduling slots (\Cref{subsec:sched_fair}).
In contrast, PF with static shadowing grows at a slower rate (slope $0.17$): the PF metric
normalizes out the fixed inverse-Gamma component, leaving only the exponential-tailed multipath fading~$g_k$
to drive user selection.
The small residual slope ($0.17$ vs.\ the theoretical $\log_2\!\log K \approx 0$) reflects the logarithmic
Gumbel growth at finite~$K$.

\vspace{-3mm}
\subsection{MIMO Random Beamforming Validation}
\vspace{-1mm}

\Cref{fig:mimo} validates the MIMO random beamforming result of \Cref{prop:mimo-rb} for $m_s \hsp{-2}{-2}{=} 2$ and $M \hsp{-2}{-2}{\in} \{1, 2, 4\}$ antennas.
The sum-rate is expected to scale as $(M/2)\log_2 K$. Inter-beam interference is \emph{not} included in the simulation
(each beam is treated as an independent SISO channel); this is justified by the concentrating-measure argument that for $K \hsp{-2}{-2}{\gg} M$,
the best user per beam is nearly orthogonal to the other beams~\cite{Sharif2005}.
The $K$ range is limited to $[10, 1000]$.

The fitted slopes are $0.52$ ($M \hsp{-2}{-2}{=} 1$; theory: $0.50$), $1.02$ ($M \hsp{-2}{-2}{=} 2$; theory: $1.00$),
and $1.92$ ($M \hsp{-2}{-2}{=} 4$; theory: $2.00$).
All values within $4\%$ of the predicted values, thus confirming the $M$-fold scaling predicted by \Cref{prop:mimo-rb}.
The slight undershoot for $M \hsp{-2}{-2}{=} 4$ ($1.92$ vs $2.00$) is consistent with finite-$K$ effects at the lower end of the
range ($K \hsp{-2}{-2}{=} 10$ with $M \hsp{-2}{-2}{=} 4$ gives only $2.5$ effective users per beam).
The close overall agreement validates the MIMO extension and demonstrates that the Fr\'echet diversity gain compounds
multiplicatively with the number of antennas.

\begin{figure}[t]
\centering
\includegraphics[width=\myfigwidth]{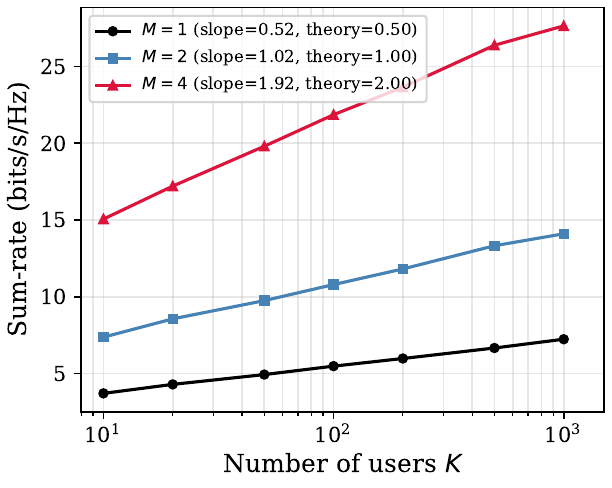}
\vspace{-2mm}
\caption{MIMO random beamforming sum-rate vs.\ $K$ for $M \in \{1, 2, 4\}$ antennas ($m_s = 2$, interference-limited, $800$--$3000$ trials per point).
The fitted slopes closely match the predicted $(M/m_s)\log_2 K$ scaling. Bootstrap $95\%$ CIs have relative width $<4\%$ for all points and are omitted for clarity.}
\label{fig:mimo}
\vspace{-3mm}
\end{figure}

\vspace{-3mm}
\section{Discussion}\label{sec:discussion}

\vspace{-1mm}
\subsection{Gumbel-Fr\'echet Phase Transition}
\label{subsec:transition}

As the shadowing parameter $m_s \hsp{-2}{-2}{\to} \infty$, the inverse-Gamma shadowing concentrates around its mean.
By Slutsky's theorem~\cite{Billingsley1999}, the composite channel converges in distribution to Nakagami-$m$ fading (Gumbel domain).
However, for any finite~$m_s$, the inverse-Gamma component retains its power-law tail, keeping the system strictly in the Fr\'echet domain.
The $\frac{1}{m_s}\log_2 K$ scaling does not degrade continuously into the Gumbel $\log_2\!\log K$ law;
the theoretical domain transition is discrete at $m_s \hsp{-2}{-2}{=} \infty$.

What changes continuously is the \emph{practical relevance} of the Fr\'echet asymptotics.
As $m_s$ grows, the power-law tail is pushed to increasingly extreme values, requiring super-exponentially more users to observe it.

\begin{proposition}[Critical user count for Fr\'echet onset]\label{prop:crossover}
  The power-law tail approximation $\Prob[h > x] \hsp{-2}{-2}{\sim} C\,x^{-m_s}$ holds for $x \hsp{-2}{-2}{\gg} x_0$,
  where $x_0 \hsp{-2}{-2}{=} (m_s - 1)/m$ is the mode of the inverse-Gamma distribution.
  The Fr\'echet normalizer $a_K \hsp{-2}{-2}{=} (AK)^{1/m_s}$ must satisfy $a_K \hsp{-2}{-2}{\ge} x_0$.
  Solving for $K$ yields the critical user count:
  \begin{equation}\label{eq:K-crossover}
    K^\star \sim \frac{1}{A}\!\left(\frac{m_s - 1}{m}\right)^{\!m_s}.
  \end{equation}
\end{proposition}

Evaluating \eqref{eq:K-crossover} (for $m \hsp{-2}{-2}{=} 1$ and normalizing $A \hsp{-2}{-2}{=} 1$ for illustration) reveals clear practical boundaries.
For heavy shadowing ($m_s \hsp{-2}{-2}{\le} 3$), $K^\star \hsp{-2}{-2}{\le} 8$, meaning Fr\'echet asymptotics are accurate almost immediately.
For $m_s \hsp{-2}{-2}{=} 5$, $K^\star \hsp{-2}{-2}{\approx} 1024$ places the crossover at the edge of practical relevance (e.g., dense stadium deployments);
the pre-asymptotic regime (often studied via ``penultimate approximations''~\cite{Resnick2007}) still contaminates much of the operating range.
Finally, for lighter shadowing ($m_s \hsp{-2}{-2}{=} 10$), $K^\star \hsp{-2}{-2}{\approx} 3.5 \hsp{-2}{-2}{\times} 10^9$ is astronomically large,
so the system behaves effectively as if in the Gumbel domain.

In summary, the polynomial diversity gain is largest in severe shadowing environments ($m_s \hsp{-2}{-2}{\le} 3$),
where baseline channel quality is poorest and diversity is most needed.
Conversely, for $m_s \hsp{-2}{-2}{\ge} 10$, Fr\'echet asymptotics are irrelevant at practical~$K$,
but the baseline capacity is already high, making the slower Gumbel scaling less consequential.

\vspace{-2mm}
\subsection{Which Shadowing Models Lead to Fr\'echet Scaling?}

The results in this work stem entirely from the inverse-Gamma shadowing's power-law tail.
\Cref{tab:scaling-comparison} classifies widely used composite models by their EVT domain and summarizes the resulting capacity scaling across representative interference scenarios.
Crucially, Gamma and lognormal shadowing preserve the slow Gumbel scaling ($\log_2\!\log_2 K$), whereas Pareto-type shadowing (e.g., inverse-Gamma)
yields Fr\'{e}chet scaling regardless of whether the system is noise- or interference-limited. The multi-path component affects only the tail constant, never the index.

\begin{table*}[t]
\centering
\caption{Tail Classification and Capacity Scaling Across Various Environments}
\label{tab:scaling-comparison}
\vspace{-2mm}
\begin{tabular}{@{}lllll@{}}
\toprule
\textbf{Channel Model} & \textbf{Shadowing (Tail)} & \textbf{Interference} & \textbf{Domain} & \textbf{Scaling} \\
\midrule
Rayleigh / Nakagami-$m$ & None (Exp.) & Noise-limited & Gumbel & $\log_2\!\log_2 K$ \\
Nak-$m$ $\times$ Gamma & Gamma (Exp.) & Noise-limited & Gumbel & $\log_2\!\log_2 K$ \\
Nak-$m$ $\times$ Lognorm. & Lognormal (Sub-exp.) & Noise-limited & Gumbel & $\log_2\!\log_2 K$ \\
\midrule
$\mathcal{F}$ / $\kappa$-$\mu$ shad. & Inv-Gamma ($x^{-m_s}$) & Noise-limited & Fr\'echet & $\tfrac{1}{m_s}\log_2 K$ \\
$\mathcal{F}(2,2m_s)$ & Inv-Gamma ($x^{-m_s}$) & Poisson PPP & Fr\'echet & $\tfrac{1}{m_s}\log_2 K$ \\
\bottomrule
\end{tabular}
\vspace{-3mm}
\end{table*}

\vspace{-3mm}
\subsection{Energy Efficiency and Power Scaling}

The heavy-tailed multi-user diversity gain also enables power reduction while maintaining quality of service.

\begin{corollary}\label{cor:power-scaling}
In the noise-limited regime, to maintain a target per-user rate $R_0$ as $K$ grows, the required transmit power $P(K)$ scales as
$P(K) \hsp{-2}{-2}{\sim} K^{-1/m_s}$, namely, the Fr\'echet regime.
\end{corollary}

\begin{proof}[Proof sketch]
In the noise-limited regime, $A \hsp{-2}{-2}{=} C_h\,(P\,r_k^{-\eta}/\sigma^2)^{m_s} \hsp{-2}{-2}{\propto} P^{m_s}$.
As $\ergcap(K) \hsp{-2}{-2}{\sim} \frac{1}{m_s}\log_2(A\,K^{1/m_s})$, maintaining $\ergcap(K) \hsp{-2}{-2}{=} R_0$ requires
$A \cdot K^{1/m_s} \hsp{-2}{-2}{=} \mathrm{const}$, hence $P \hsp{-2}{-2}{\sim} K^{-1/m_s}$.
\end{proof}

\begin{remark}\label{rem:green}
In the Gumbel regime, $P(K) \sim 1/\!\log K$, only a logarithmic power reduction.
For $m_s = 2$, the Fr\'echet result gives $P(K) \sim 1/\!\sqrt{K}$: doubling the user count cuts
transmit power by approximately $30\%$. This polynomial power back-off is significant for energy-efficient
network design and can have direct implications for battery lifetime in IoT deployments.
In the interference-limited regime, this power-scaling interpretation does not apply directly:
since all transmitters share the same power~$P$, the aggregate interference $I_k \propto P$,
and the SIR $= P\,r_k^{-\eta}\,h_k / I_k$ is invariant to~$P$.
Power reduction in this regime requires heterogeneous power control or coordinated schemes
that break the proportionality between signal and interference power.
\end{remark}

\vspace{-3mm}
\subsection{Limitations and Future Work}

Several limitations of the current analysis provide avenues for future work.
First, the i.i.d.\ shadowing assumption does not capture spatial correlation among nearby users.
When users cluster into $G$ correlation groups sharing identical shadowing, the effective diversity pool
reduces to $K/G$ independent samples; the $1/m_s$ exponent is preserved but the constant~$A$ decreases,
requiring proportionally more users to achieve a given capacity target.
Quantifying this ``effective $K$'' under realistic correlation models merits further investigation.
Second, opportunistic scheduling is sensitive to feedback delays and estimation errors. Fortunately, because Fr\'echet scaling
is driven by the slow-varying shadowing component~$s_k$, it remains robust to stale estimates; fast-fading errors reduce the
scaling constant~$A$ but preserve the polynomial exponent~$1/m_s$.
However, severely limited feedback (e.g., 1-bit thresholding) restricts the effective user pool, potentially degrading the scaling law.
Notably, threshold-based feedback is particularly well-suited to the Fr\'echet domain: the power-law tail
means that extreme users exceed any fixed threshold by a wide margin, so even coarse feedback
reliably identifies the best user.

Regarding network geometry, realistic base station spatial repulsion (e.g., Ginibre or Mat\'ern processes~\cite{Deng2015Ginibre,Miyoshi2014})
yields lighter-tailed interference than the PPP. Since the fading distribution dictates the SINR tail (see \Cref{thm:sinr-tail}),
the polynomial scaling remains theoretically robust.
Practically, its onset depends on~$m_s$: polynomial growth accurately captures moderate~$K$ when $m_s \hsp{-2}{-2}{\le} 3$,
while classical Gumbel approximations suit mild shadowing ($m_s \hsp{-2}{-2}{\ge} 5$).
Finally, temporal correlation in slow-fading reduces the effective diversity order,
requiring the scaling to be evaluated relative to the number of independent coherence blocks.

\vspace{-3mm}
\section{Conclusion}\label{sec:conclusion}

In this work, we demonstrated that the classical $\log_2\!\log K$ multi-user diversity scaling law is an artifact of
exponential-tailed (Rayleigh/Nakagami) fading models. By modeling composite fading via the Fisher--Snedecor $\mathcal{F}$ distribution
(which accounts for both Nakagami-$m$ multi-path fading and inverse-Gamma shadowing), we showed that the channel power exhibits a power-law tail with index $m_s$.
In interference-limited Poisson networks, the SINR inherits this heavy-tailed property, causing the maximum SINR among $K$ users to fall into
the Fr\'{e}chet domain of attraction.

The resulting scaling is polynomial, defined by $\gamma_{\max}^{(K)} \hsp{-2}{-2}{\sim} K^{1/m_s}$, which yields an ergodic capacity scaling of $\frac{1}{m_s}\log_2 K$.
For the canonical case of $m_s \hsp{-2}{-2}{=} 2$, this translates to a $\sqrt{K}$ SINR scaling and $\frac{1}{2}\log_2 K$ capacity growth.
These findings suggest that multi-user diversity is a significantly more potent resource in dense, interference-limited networks than previously appreciated,
provided the scheduler is designed to exploit heavy-tailed fluctuations.

Numerical results confirm our analytical predictions across user counts from $K \hsp{-2}{-2}{=} 10$ to $5000$,
validating the Fr\'{e}chet scaling and convergence. Furthermore, our extensions provide a robust toolkit for next-generation network design.
The MIMO random beamforming analysis (\Cref{prop:mimo-rb}) establishes a sum-rate scaling of $(M/m_s)\log_2 K$,
while the energy efficiency results (\Cref{cor:power-scaling}) show that transmit power can be reduced polynomially as $K^{-1/m_s}$ -
a more aggressive reduction than the $1/\ln K$ back-off possible in the Gumbel regime.
Coupled with the crossover criteria in \Cref{subsec:transition}, these results offer a comprehensive framework for harnessing heavy-tailed multi-user diversity in future deployments.

\vspace{-2mm}
\bibliographystyle{unsrt}
\bibliography{references}

@article{Suzuki1977,
  author  = {H. Suzuki},
  title   = {A statistical model for urban radio propagation},
  journal = {IEEE Trans. Commun.},
  volume  = {25},
  number  = {7},
  pages   = {673--680},
  month   = {Jul.},
  year    = {1977}
}

@inproceedings{Knopp1995,
  author    = {R. Knopp and P. A. Humblet},
  title     = {Information capacity and power control in single-cell multiuser communications},
  booktitle = {Proc. IEEE ICC},
  address   = {Seattle, WA},
  month     = {Jun.},
  year      = {1995},
  pages     = {331--335}
}

@article{Viswanath2002,
  author  = {P. Viswanath and D. N. C. Tse and R. Laroia},
  title   = {Opportunistic beamforming using dumb antennas},
  journal = {IEEE Trans. Inf. Theory},
  volume  = {48},
  number  = {6},
  pages   = {1277--1294},
  month   = {Jun.},
  year    = {2002}
}

@article{Sharif2005,
  author  = {M. Sharif and B. Hassibi},
  title   = {On the capacity of {MIMO} broadcast channels with partial side information},
  journal = {IEEE Trans. Inf. Theory},
  volume  = {51},
  number  = {2},
  pages   = {506--522},
  month   = {Feb.},
  year    = {2005}
}

@article{Sharif2007,
  author  = {M. Sharif and B. Hassibi},
  title   = {A comparison of time-sharing, {DPC}, and beamforming for {MIMO} broadcast channels with many users},
  journal = {IEEE Trans. Commun.},
  volume  = {55},
  number  = {1},
  pages   = {11--15},
  month   = {Jan.},
  year    = {2007}
}

@article{Yoo2017,
  author  = {S. K. Yoo and S. L. Cotton and P. C. Sofotasios and M. Matthaiou and M. Valkama and G. K. Karagiannidis},
  title   = {The {Fisher}--{Snedecor} $\mathcal{F}$ distribution: A simple and accurate composite fading model},
  journal = {IEEE Commun. Lett.},
  volume  = {21},
  number  = {7},
  pages   = {1661--1664},
  month   = {Jul.},
  year    = {2017}
}

@article{Badarneh2018,
  author  = {O. S. Badarneh and D. B. da Costa and P. C. Sofotasios and S. Muhaidat and S. L. Cotton},
  title   = {On the sum of {Fisher}--{Snedecor} $\mathcal{F}$ variates and its application to maximal-ratio combining},
  journal = {IEEE Wireless Commun. Lett.},
  volume  = {7},
  number  = {6},
  pages   = {966--969},
  month   = {Dec.},
  year    = {2018}
}

@article{Cotton2014,
  author  = {S. L. Cotton},
  title   = {Human body shadowing in cellular device-to-device communications: Channel modeling using the shadowed $\kappa$--$\mu$ fading model},
  journal = {IEEE J. Sel. Areas Commun.},
  volume  = {33},
  number  = {1},
  pages   = {111--119},
  month   = {Jan.},
  year    = {2015}
}

@article{Andrews2011,
  author  = {J. G. Andrews and F. Baccelli and R. K. Ganti},
  title   = {A tractable approach to coverage and rate in cellular networks},
  journal = {IEEE Trans. Commun.},
  volume  = {59},
  number  = {11},
  pages   = {3122--3134},
  month   = {Nov.},
  year    = {2011}
}

@article{Baccelli2009,
  author  = {F. Baccelli and B. B\l{}aszczyszyn},
  title   = {Stochastic geometry and wireless networks, Volume I: Theory},
  journal = {Found. Trends Netw.},
  volume  = {3},
  number  = {3--4},
  pages   = {249--449},
  year    = {2009}
}

@article{Haenggi2009,
  author  = {M. Haenggi and J. G. Andrews and F. Baccelli and O. Dousse and M. Franceschetti},
  title   = {Stochastic geometry and random graphs for the analysis and design of wireless networks},
  journal = {IEEE J. Sel. Areas Commun.},
  volume  = {27},
  number  = {7},
  pages   = {1029--1046},
  month   = {Sep.},
  year    = {2009}
}

@article{Win2009,
  author  = {M. Z. Win and P. C. Pinto and L. A. Shepp},
  title   = {A mathematical theory of network interference and its applications},
  journal = {Proc. IEEE},
  volume  = {97},
  number  = {2},
  pages   = {205--230},
  month   = {Feb.},
  year    = {2009}
}

@article{ElSawy2017,
  author  = {H. ElSawy and A. Sultan-Salem and M. S. Alouini and M. Z. Win},
  title   = {Modeling and analysis of cellular networks using stochastic geometry: A tutorial},
  journal = {IEEE Commun. Surveys Tuts.},
  volume  = {19},
  number  = {1},
  pages   = {167--203},
  year    = {2017}
}

@article{Schilcher2016,
  author  = {U. Schilcher and C. Bettstetter and M. Haenggi},
  title   = {Interference functionals in {Poisson} networks},
  journal = {IEEE Trans. Inf. Theory},
  volume  = {62},
  number  = {1},
  pages   = {370--383},
  month   = {Jan.},
  year    = {2016}
}

@book{DeHaan2006,
  author    = {L. de Haan and A. Ferreira},
  title     = {Extreme Value Theory: An Introduction},
  publisher = {Springer},
  address   = {New York, NY},
  year      = {2006}
}

@book{Resnick2007,
  author    = {S. I. Resnick},
  title     = {Heavy-Tail Phenomena: Probabilistic and Statistical Modeling},
  publisher = {Springer},
  address   = {New York, NY},
  year      = {2007}
}

@article{Breiman1965,
  author  = {L. Breiman},
  title   = {On some limit theorems similar to the arc-sin law},
  journal = {Theory Probab. Appl.},
  volume  = {10},
  number  = {2},
  pages   = {323--331},
  year    = {1965}
}

@article{AlBadarneh2018,
  author  = {Y. Al-Badarneh and C. N. Georghiades and S. Alhussein},
  title   = {Asymptotic performance analysis of the $k$-th best link selection over wireless fading channels: An extreme value theory approach},
  journal = {IEEE Trans. Veh. Technol.},
  volume  = {67},
  number  = {7},
  pages   = {6652--6657},
  month   = {Jul.},
  year    = {2018}
}

@article{AlAhmadi2010,
  author  = {A. S. Al-Ahmadi and H. Yanikomeroglu},
  title   = {On the approximation of the generalized-$K$ distribution by a gamma distribution for modeling composite fading channels},
  journal = {IEEE Trans. Wireless Commun.},
  volume  = {9},
  number  = {2},
  pages   = {706--713},
  month   = {Feb.},
  year    = {2010}
}

@article{Alouini1999,
  author  = {M.-S. Alouini and A. J. Goldsmith},
  title   = {Capacity of {Rayleigh} fading channels under different adaptive transmission and diversity-combining techniques},
  journal = {IEEE Trans. Veh. Technol.},
  volume  = {48},
  number  = {4},
  pages   = {1165--1181},
  month   = {Jul.},
  year    = {1999}
}

@article{CMS1976,
  author  = {J. M. Chambers and C. L. Mallows and B. W. Stuck},
  title   = {A method for simulating stable random variables},
  journal = {J. Amer. Statist. Assoc.},
  volume  = {71},
  number  = {354},
  pages   = {340--344},
  month   = {Jun.},
  year    = {1976}
}

@book{Billingsley1999,
  author    = {P. Billingsley},
  title     = {Convergence of Probability Measures},
  edition   = {2nd},
  publisher = {Wiley},
  address   = {New York, NY},
  year      = {1999}
}

@book{Leadbetter1983,
  author    = {M. R. Leadbetter and G. Lindgren and H. Rootz\'{e}n},
  title     = {Extremes and Related Properties of Random Sequences and Processes},
  publisher = {Springer},
  address   = {New York, NY},
  year      = {1983}
}

@book{Embrechts1997,
  author    = {P. Embrechts and C. Kl\"{u}ppelberg and T. Mikosch},
  title     = {Modelling Extremal Events for Insurance and Finance},
  publisher = {Springer},
  address   = {Berlin},
  year      = {1997}
}

@book{Haenggi2012,
  author    = {M. Haenggi},
  title     = {Stochastic Geometry for Wireless Networks},
  publisher = {Cambridge University Press},
  address   = {Cambridge, UK},
  year      = {2012}
}

@article{Dhillon2012,
  author  = {H. S. Dhillon and R. K. Ganti and F. Baccelli and J. G. Andrews},
  title   = {Modeling and analysis of {K}-tier downlink heterogeneous cellular networks},
  journal = {IEEE J. Sel. Areas Commun.},
  volume  = {30},
  number  = {3},
  pages   = {550--560},
  month   = {Apr.},
  year    = {2012}
}

@article{Gesbert2004,
  author  = {D. Gesbert and M.-S. Alouini},
  title   = {How much feedback is multi-user diversity really worth?},
  journal = {Proc. IEEE Int. Conf. Commun. (ICC)},
  pages   = {234--238},
  month   = {Jun.},
  year    = {2004}
}

@article{Yoo2019,
  author  = {S. K. Yoo and P. C. Sofotasios and S. L. Cotton and S. Muhaidat and F. J. Lopez-Martinez and J. M. Romero-Jerez and G. K. Karagiannidis},
  title   = {A comprehensive analysis of the achievable channel capacity in $\mathcal{F}$ composite fading channels},
  journal = {IEEE Access},
  volume  = {7},
  pages   = {34078--34094},
  year    = {2019}
}

@article{Miyoshi2014,
  author  = {N. Miyoshi and T. Shirai},
  title   = {A cellular network model with {Ginibre} configured base stations},
  journal = {Adv. Appl. Probab.},
  volume  = {46},
  number  = {3},
  pages   = {832--845},
  year    = {2014}
}

@article{Deng2015Ginibre,
  author  = {N. Deng and W. Zhou and M. Haenggi},
  title   = {The {Ginibre} point process as a model for wireless networks with repulsion},
  journal = {IEEE Trans. Wireless Commun.},
  volume  = {14},
  number  = {1},
  pages   = {107--121},
  month   = {Jan.},
  year    = {2015}
}

\end{document}